\documentclass{llncs}
\usepackage{makeidx}  
\usepackage{graphicx}
\usepackage{amsmath}
\usepackage{amssymb}
\usepackage{xparse}
\usepackage{xspace}
\usepackage[T1]{fontenc} 
\usepackage{xcolor}
\usepackage{subcaption}
\graphicspath{{pics/}}
\NewDocumentCommand{\ipic}{mO{scale=0.4}}{
	\includegraphics[#2]{#1.pdf}
}

\NewDocumentCommand{\pic}{O{}mO{scale=0.4}}{
	\begin{figure}[ht]
		\center
		\ipic{#2}[#3]
		\caption{\label{#2}	#1}
	\end{figure}
}

\NewDocumentCommand{\pics}{omo}{
	\begin{figure}[t]
		\center
		#2
		\IfValueTF{#1}{\caption{#1 
			\IfValueTF{#3}{\label{#3}}{}
		}}{}		
	\end{figure}
}

\NewDocumentCommand{\subpic}{O{scale=0.4}mm}{
	\subcaptionbox{#2 
		\label{#3}		
	}{\ipic{#3}[#1]}
}

\NewDocumentCommand{\sym}{smO{}moo}{
	\expandafter\NewDocumentCommand\csname #2\endcsname{#3}{\ensuremath{{#4}}\xspace}
}
\NewDocumentCommand{\name}{mmm}{
	\expandafter\NewDocumentCommand\csname #1n\endcsname{}{#2\xspace}
	\expandafter\NewDocumentCommand\csname #1a\endcsname{}{#3\xspace}	
	\expandafter\NewDocumentCommand\csname #1ns\endcsname{}{#2s\xspace}
	\expandafter\NewDocumentCommand\csname #1as\endcsname{}{#3s\xspace}	
}
\sym{func}[mmm]{#1: #2 \rightarrow #3}
\sym{powerset}[mO{}]{\mathcal{P}_{#2}(#1)}
\sym{reals}{\bbbr}[real numbers]
\sym{preals}{\reals^{+}}[positive real numbers]
\sym{intervals}[O{\reals}]{\mathcal{I}(#1)}
\sym{intervalsp}[O{\preals}]{\intervals[\preals]}[set of all intervalls where the lower bound is positive]
\sym{interval}[O{x}O{y}]{[#1, #2]}[an interval where $x$ and $y$ are mostly in \reals]
\sym{naturalnumbers}{\bbbn}[natural numbers]

\sym{buses}[O{}O{}]{N_{#1}^{#2}}[$= \{1, \ldots n\}$, set of buses]
\sym{numofbuses}{n}[number of buses, \cardi{\buses}]

\sym{gens}{G}
\sym{gen}[O{}]{\gens(#1)}
\sym{generators}[O{}]{\buses_{\gens}^{#1}}

\sym{loadsy}{L}
\sym{load}[O{}]{\loadsy(#1)}
\sym{loads}[O{}]{\buses_{\loadsy}^{#1}}

\sym{pas}{\Theta}
\sym{pa}[O{}O{}]{\pas({#1})^{#2}}
\name{pa}{phase angle}{PA}

\sym{lines}[O{}O{}]{E_{#1}^{#2}}
\sym{linesw}[O{}]{\lines_{#1}'}

\sym{suses}{S}
\sym{sus}[O{}O{}O{}O{}]{\suses(\edge[#1][#2][#3][#4])}
\sym{suscaps}{\suses^*}
\sym{suscap}[O{}O{}O{}O{}]{\suscaps(\edge[#1][#2][#3][#4])}

\sym{edge}[O{a}O{b}O{}O{}]{\{#1, #2\}_{#3}^{#4}}
\sym{arc}[O{}]{e_{#1}}

\sym{flows}{F}
\sym{flow}[O{}O{}O{}O{}]{\flows(\edge[#1][#2][#3][#4])}
\sym{capas}{C}
\sym{capa}[O{}O{}O{}O{}]{\capas(\edge[#1][#2][#3][#4])}

\sym{net}[O{}O{}]{\mathcal{N}_{#1}^{#2}}[network]

\sym{mpf}[O{\net}O{}]{\mathcal{MPF}_{#2}(#1)}
\name{mpf}{maximum potential flow}{MPF}

\sym{mff}[O{\net}O{}]{\mathcal{MFF}_{#2}(#1)}
\name{mff}{maximum FACTS flow}{MFF}

\sym{msf}[O{\net}O{}]{\mathcal{MSF}_{#2}(#1)}
\name{msf}{maximum switching flow}{MSF}

\name{dc}{Linear DC network}{LDC network}
\name{fdc}{FACTS Linear DC network}{FLDC network}

\sym{gfch}[O{}O{}]{\mathit{GFCN}_{#1}^{#2}}
\name{gfch}{generation-FACTS-choice network}{GFCN}

\sym{gsch}[O{}O{}]{\mathit{GSCN}_{#1}^{#2}}
\name{gsch}{generation-Switching-choice network}{GSCN}

\sym{lfch}[O{}O{}]{\mathit{LFCN}_{#1}^{#2}}
\name{lfch}{load-FACTS-choice network}{LFCN}

\sym{lsch}[O{}O{}]{\mathit{LSCN}_{#1}^{#2}}
\name{lsch}{load-Switching-choice network}{LSCN}

\sym{sols}{\mathcal{T}_{\net}}
\sym{sol}{(\suses, \pas, \flows, \gens, \loadsy)}
\sym{dsol}{(\pas, \flows, \gens, \loadsy)}
\sym{swsol}{(\linesw, \pas, \flows, \gens, \loadsy)}
\name{sol}{feasible solution}{FS}
\name{optsol}{optimal solution}{OS}

\begin{document}
\title{The Complexity of Switching and FACTS Maximum-Potential-Flow Problems}
\author{Karsten Lehmann\inst{2,1} \and Alban Grastien\inst{2,1} \and Pascal Van Hentenryck\inst{1,2}}
\institute{Artificial Intelligence Group, Australian National University\\
    \and
    Optimisation Research Group, National ICT Australia, Canberra Laboratory
    \email{\{first.last\}@nicta.com.au}
} 
\maketitle              

\begin{abstract}
This papers considers the problem of maximizing 
the load that can be served by a power network.  
We use the commonly accepted Linear DC power network model 
and consider two configuration options: 
switching lines and using FACTS devices.  
We present the first comprehensive complexity study of this optimization problem.  
Our results show that the problem is NP-complete and that there is no fully polynomial-time approximation scheme. 
For switching, these results extend to planar networks with a maximum-node degree of $3$.
Additionally, we demonstrate that the optimization problems are still NP-hard if we restrict the network structure to cacti with a maximum degree of $3$.

\end{abstract}
\section{Introduction}

By design, power networks are capable of satisfying a larger demand
than normal operations require.  In critical situations however, e.g.,
when heat waves cause the demand to peak or when a natural disaster
causes significant damage to the network, the full demand becomes
impossible to meet and power utilities face the problem of
maximizing the power they can deliver to customers.

For instance, \emph{power system restoration} following a disaster is
the problem of repairing and reconfiguring the network to resupply as
much of the demand as possible as fast as possible.  This problem has
been studied by power engineers for over 30 years (see
\cite{adibi2000power} for a comprehensive collection of works).
Finding the network configuration that maximizes the load served is a
subproblem of power system restoration: the latter requires the
configuration to be optimized after each repair action and additionally
needs to integrate issues such as the routing of repair crews \cite{PSCC11}.
Because it is only a part of this more complex optimization problem,
finding the optimal configuration -- or at
least providing guarantees on the quality of the configuration obtained -- must be achieved quickly.

This paper presents the first comprehensive study of the computational
complexity of maximizing the flow in a power network.  We use the
commonly accepted Linear DC power network model and consider two
reconfiguration options: reconfiguration via {\em line switching} and 
reconfiguration via the use of \emph{Flexible
  Alternating Current Transmission System} (FACTS) devices.
Line switching means physically disconnecting 
two nodes that were previously connected 
(or vice versa, connecting two
previously disconnected nodes).
Power networks exhibit a phenomenon 
akin to Braess Paradox in transport networks 
\cite{Braess_1968_UbereinParadoxon} 
whereby switching lines off can increase the maximum power flow.  
This is because the flow in a power line depends 
on the difference of potentials of its end nodes:
this creates additional (cyclic) constraints on the flow 
that can be eliminated by switching.
FACTS devices can modify the physical characteristics of the lines -- 
in particular their susceptance --
and partially relax the cyclic constraints.  
The switching reconfiguration problem consists in choosing the set of
lines that should be switched off, whereas the FACTS reconfiguration
problem consists in choosing the parameters of the lines (given
appropriate upper and lower bounds for these parameters).
In both cases the objective is to maximize the power flow through the
network.

In this paper, we present complexity results answering the question of
whether there exist efficient algorithms to solve the switching and
FACTS max-flow problems. The detailed contributions of the paper and
its organization are as follows.  We first describe the problems in their
full mathematical details (Section~\ref{model}).  Then in
Section~\ref{sec:general_complexity}, we present general complexity results
showing that the problems do not admit a 
fully polynomial time approximation scheme.  Real power
networks usually have a simple structure: they are almost planar and
have a fixed maximum degree.  To obtain complexity results that are
relevant to the real-world case, we search for the most ``complex'' network structure 
for which the reconfiguration problems are easy.  In Section~\ref{sec:beyond_trees} we
show that \emph{cacti} with a fixed maximum degree of $3$ already make
the problems NP-hard.  We prove a similar result for another simple
relaxation of trees, called \emph{n-level tree networks}.  This leads
us to conclude that trees are the most ``complex'' network structure
that can be easily configured optimally, although whether there exists a
fully polynomial time approximation scheme for cacti and/or n-level
tree networks is still an open question.
It should be noted that the complexity results 
we present use unrealistic network parameters for simplicity.  
The values can all be scaled to be realistic without influencing the results.
The body of the paper presents the reductions needed to establish our complexity results; the proofs are given in appendix.  


\section{Problem Definition}
\label{model}

This section presents the network model and power flow equations we use, and defines the reconfiguration problems we study.

\subsection{Network Model}

In this paper we use the \emph{Linear DC (LDC) model} of electrical power networks \cite{Schweppe_1970_Powersystemstatic}.
The LDC model is a linearization of the nonlinear steady-state electrical power flow equations (Alternating Current Model) and is widely used in practice.  
It assumes that all voltage magnitudes are one in the per-unit system and ignores reactive power and resistance which are small relative to real power and reactance during normal operations. 
What is left is the susceptance%
\footnote{Susceptance is a negative value but when used to calculate the flow it is multiplied with $-1$. For readability, we omit the $-1$ and make the susceptance a positive value.}
(the inverse of the reactance), the capacity and the phase angles of the voltages.
The flow that is transmitted by an edge is the product of the phase angle difference between its two ends and the susceptance.  
The network may be equipped with FACTS devices, which are
physical devices allowing the (otherwise constant) susceptance parameter to vary.

In the following,
\powerset{X}[2] denotes the set of all subsets of $X$ with 2 elements. We write \intervals (resp. \intervalsp) for the set of all intervals 
over the set $\reals$ (resp. $\preals$). 


\begin{definition}
	A \emph{\fdcn (\fdca)} is a tuple $\net = (\buses, \generators, \loads, \lines)$ where \buses is the set of \emph{nodes}; $\generators \subseteq \buses$ is the set of \emph{generators}; $\loads \subseteq \buses \setminus \generators$ is the set of \emph{loads}; and $\lines \subseteq \powerset{\buses}[2] \times \intervalsp \times \preals$ 
is the set of \emph{edges} with their \emph{susceptance limits} and \emph{capacity} 
such that no two edges connect the same pair of nodes, i.e., $(\{a,b\}, sl_1, c_1), (\{a,b\}, sl_2, c_2) \in \lines \implies sl_1 = sl_2 \text{ and } c_1 = c_2$.  
\end{definition}

We define functions \func{\suscaps}{\lines}{\intervalsp} and \func{\capas}{\lines}{\preals} such that $\suscaps(e)$ is the susceptance limit of an edge and $\capas(e)$ is its capacity.
Moreover, we write $\edge[a][b][[s,t]][p]$ for an edge from $a$ to $b$ with susceptance limit $[s,t]$ and capacity $p$.
If the susceptance is fixed, i.e., $t = s$, we shall write \edge[a][b][s][p]. We may also ignore these values and simply refer to the edge by \edge[a][b].
The above model does not explicitly give upper bounds on the generation or load of a node.  
Such constraints can be modeled by connecting this node to the network 
through a single edge whose capacity is the maximum output/intake of the node.  

An \fdca where all susceptances are fixed, i.e., without
FACTS devices, is called a \dcn.  The only reconfiguration option for
such networks is switching.

\begin{definition}
    An \fdca \net is called a \emph{\dcn (\dca)} if all susceptances are fixed, i.e., $\forall \edge[a][b][[s,t]][p] \in \lines: s = t$.
\end{definition}
To represent switched network configurations, we need the notion of a subnetwork.
The \emph{sub-network} $\net[\linesw]$ of an \fdca $\net = (\buses, \generators, \loads, \lines)$ in which the edges in \linesw have been removed is $\net[\linesw] := (\buses, \generators, \loads,  \lines \setminus \linesw).$
The \emph{sum} $\net + \net'$ of two \fdcas $\net = (\buses, \generators, \loads, \lines)$ and $\net' = (\buses', \generators', \loads', \lines')$ with $\lines \cap \lines' = \emptyset$ is the \fdca $\net + \net' := (\buses \cup \buses', \generators \cup \generators', \loads \cup \loads', \lines \cup \lines')$.

\subsection{Power Flow Equations}  

We now introduce the notations and equations pertaining to \fdca power flows.
The following definitions assume a fixed \fdca $\net = (\buses, \generators, \loads, \lines)$.
We define functions \func{\suses}{\lines}{\preals} and \func{\pas}{\buses}{\reals} such that $\suses(e)$ is the \emph{susceptance} of edge $e$ and $\pas(a)$ is the 
\emph{\pan} at node $a$. Moreover the \emph{generation} and \emph{load} at a node are given by functions \func{\gens}{\buses}{\preals} and \func{\loadsy}{\buses}{\preals} such that $\forall a \in \buses \setminus \generators: \gen[a] := 0$ and $\forall a \in \buses \setminus \loads: \load[a] := 0$. 

The \emph{flow} on an edge is given by function \func{\flows}{\lines}{\reals}. 
The edges of \fdcas are undirected, however to be able to describe flows, we need a notion of orientation of an edge but the concrete orientation we choose does not influence the theory.
To that end, whenever we define an edge, we abuse the notation \edge[a][b] to also represent that $\flow[a][b] \geq 0$ whenever the flow goes from $a$ to $b$ and $\flow[a][b] \leq 0$ otherwise.

The \dca model imposes two sets of constraints: Kirchhoff's conservation law and the \dca power law.
Kirchhoff's conservation law states that the power that enters a node equals the power that leaves this node.  
The \dca power law binds together the power flow, the phase angle and the susceptance of an edge.  

\begin{definition}
  A triple $(\flows, \gens, \loadsy)$ satisfies \emph{Kirchhoff's conservation law} if $\forall a \in \buses$:
  $\sum_{\edge[a][b] \in \lines} \flow[a][b] - \sum_{\edge[b][a] \in \lines} \flow[b][a] = \gen[a] - \load[a].$
\end{definition}

\begin{definition}
    A triple $(\suses, \pas, \flows)$ satisfies the \emph{\dca power law} if $\forall \edge[a][b] \in \lines: \flow[a][b] = \sus[a][b] (\pa[b] - \pa[a]).$
\end{definition}
\begin{definition}
We call a tuple \sol a \emph{\soln} if: $(\flows, \gens, \loadsy)$ satisfies Kirchhoff's conservation law; $(\suses, \pas, \flows)$ satisfies the \dca power law; $\forall \arc \in \lines: \suses(\arc) \in \suscaps(\arc)$ and $\forall \arc \in \lines: |\flows(\arc)| \leq \capas(\arc)$.\\
We write \sols for the set of all \solns of \net.
\end{definition}
For \solns of an \dca, we omit the susceptance and write \dsol.

\subsection{Reconfiguration Problems}
As mentioned in the introduction of this paper, several significant real-world applications such as load shedding and power supply restoration motivate the problem
of maximising the flow through the network.
We now define several useful variants of this problem: the \emph{\mffn (\mffa)} has to satisfy both power laws; the \emph{\mpfn (\mpfa)} is a variant of \mffa where susceptance is fixed (i.e., it applies to an \dca); finally the \emph{\msfn (\msfa)} also applies to \dcas but allows switching, i.e. the removal of edges.  


\begin{definition}
The \emph{\mpfn (\mpfa)} of an \dca \net is defined as
$\mpf := \max_{\dsol \in \sols} \sum_{g \in \buses} \gen[g].$
    A \soln \dsol with $\mpf = \sum_{g \in \buses} \gen[g]$ is called \emph{\optsoln}.

The \emph{\mffn (\mffa)} of an \fdca \net is defined as
$\mff := \max_{\sol \in \sols} \sum_{g \in \buses} \gen[g].$
    A \soln \sol with $\mff = \sum_{g \in \buses} \gen[g]$ is called \emph{\optsoln}.

The \emph{\msfn (\msfa)} is the \mpfa of a sub-network of \net that maximizes the flow, i.e.,
$\msf := \max_{\linesw \subseteq \lines} \mpf[\net[\linesw]].$
    We call \swsol an \optsoln if $\msf = \mpf[\net[\linesw]]$ and \dsol is an \optsoln for \mpf[\net[\linesw]].
\end{definition}
To establish our complexity results, we define a decision version of our optimization problems.
\begin{definition}
Given an $x \in \preals$, 
the \emph{\mpfa/\msfa/\mffa problem} is the problem of deciding if $\mpf/\msf/\mff \geq x$.
\end{definition}

\paragraph{Examples and Graphical Representation:}
\pics[Examples for \mpfa, \msfa and \mffa.]{
	\subpic{An \dca.}{triangle}
    \subpic{The max flow.}{triangle_mf}
	\subpic{The \mpfa.}{triangle_mpf}
	\subpic{The \msfa.}{triangle_msf}	
	\subpic{An \fdca.}{triangle_facts}
	\subpic{The \mffa.}{triangle_mff}		
}[pic:examples]
Figure~\ref{pic:examples} introduces our graphical representations for LDC and \fdcas along with examples for \mpfa, \msfa and \mffa. 
We omit the susceptance and/or capacity of an edge when its value is $1$.
Figure~\ref{triangle} shows an \dca where $g$ is a generator (box), $l$ is a load (house) and $b$ is a normal node (sphere). 
Its easy to see that the traditional max flow for this network is $34$ whereas in the LDC model, we only can supply $16$ as shown in Figure~\ref{triangle_mpf} because the congestion of the edge \edge[b][l][1][4] constrains the phase angle (written as $A=$ in the nodes) between $g$ and $l$. 
However, by switching the edge \edge[g][b][1][5], we can improve the maximum generation to $30$ as shown in Figure~\ref{triangle_msf}.
Figure~\ref{triangle_facts} shows a variant of the network with two FACTS devices. These allow the maximum generation to reach $28$ as shown in Figure~\ref{triangle_mff}.

\section{General Complexity}
\label{sec:general_complexity}
Solving the \mpfa problem is known to be polynomial as all constraints are linear and the problem can therefore be formulated as a linear program (LP).
In this section we demonstrate that the problem becomes harder when reconfiguration is allowed.  
To this end we build generator choice networks, which are building blocks for our proofs.  
We then show NP-Completeness in general networks and prove strong NP-hardness for switching in planar networks.  

It is easy to see that both problems are not worse than NP, since checking a given solution is as simple as summing up the generation at all buses and doing one comparison.
The \msfa optimization problem can also be formulated as a mixed-integer-linear program (MILP)\cite{Fisher_2008_OptimalTransmissionSwitching}.
\begin{lemma}
    \label{lem:innp}
The \msfa and the \mffa problem are in NP.
\end{lemma}
NP-hard problems are about (discrete) choices: for example, the choice between true or false for a SAT variable.
To prove NP-hardness, we construct sub-networks that represent such choices, 
that we call \emph{generator choice networks} (GCN). 
A GCN can be regarded as a black box with a port $v$ connected to the rest of the network.  
Given an $x \in \preals$, the power generation within the GCN will be maximal 
iff $v$ acts as generator (for the rest of the network)
that provides a power of exactly $0$ or $x$: 
the decision at each black box will therefore be 
which one of these two values the box provides.  
Each optimization problem has its own GCN: 
the \gschn for the \msfa problem and the \gfchn for the \mffa problem.  
Both types are defined below and Table~\ref{tab:cn} displays both of them and their graphical representation.

\begin{table}[h]
    \caption{\label{tab:cn} The networks \gscha and \gfcha.  
    Any connection from the rest of the network to these networks is through $v$ 
    where $v$ is no longer a load.}
    \centering
    \begin{tabular}{|c|c|c|c|}
        \hline 
        & Graph. Rep. & \gsch[x,v][-] & \gfch[x,v][-] \\
        \hline
        \raisebox{4.3\height}{Network}  & \raisebox{1.5\height}{\ipic{grcn}} & \ipic{gsch} & \ipic{gfch} \\ 
        \hline
        \raisebox{5.3\height}{$L(v) = 0$} & \raisebox{2\height}{\ipic{grcn0}} & \ipic{gsch_msf_1} & \ipic{gfch_mff0} \\ 
        \hline 
        \raisebox{5.3\height}{$L(v) = x$} & \raisebox{1.9\height}{\ipic{grcnx}} & \ipic{gsch_msf_2} & \ipic{gfch_mff1} \\ 
        \hline 
    \end{tabular} 
\end{table}

\begin{definition}[\gschn, \gscha]
Let $x \in \preals$. The \emph{\gschn for $x$ with connector $v$} $\gsch[x,v] := (\buses, \generators, \loads, \lines)$ is an \dca defined by $\generators := \{g\}$, $\loads := \{l\}$, $\buses := \generators \cup \loads \cup \{v\}$ and $\lines := \{ \edge[g][v][1][x], \edge[g][l][1][2x], \edge[v][l][1][x]\}$. 
    Let \gsch[x,v][+] be the version of \gsch[x,v] where $v$ is a generator and \gsch[x,v][-] the version where $v$ is a load.
\end{definition}


\begin{definition}[\gfchn, \gfcha] 
    Let $x \in \preals$.
    The \emph{\gfchn for $x$ with connector $v$} $\gfch[x,v] := (\buses, \generators, \loads, \lines)$  is an \fdca defined by $\generators := \{g, e, t\}$, $\loads := \{l\}$, $\buses := \generators \cup \loads \cup \{v,c\}$ and 
    $\lines := \{\edge[g][v][1][x], \edge[e][v][[0.4,1.6]][0.4x], \edge[e][c][1][0.65x],\edge[v][c][1][0.9x],\allowbreak \edge[t][c][1][x],\allowbreak \edge[t][l][1][3.55x], \edge[c][l][1][2.55x]\}.$
    Let \gfch[x,v][+] be the version of \gfch[x,v] where $v$ is a generator and \gfch[x,v][-] the version where $v$ is a load.
\end{definition}


\begin{lemma}
\label{lem:gsch}
Let $x \in \preals$ and \gsch[x,v] be the \gscha. We have 
\begin{enumerate}
    \item $\msf[\gsch[x,v][-]] = 3x$;
    \item $\{ \load[v] \mid \swsol \text{ is an \optsoln of } \msf[\gsch[x,v][-]]\} = \{0, x\}$;
    \item for every \optsoln of \msf[\gsch[x,v][+]] we have $\gen[v] > 0 \implies \gen[g] < 3x$.
\end{enumerate}
\end{lemma}

\begin{lemma}
\label{lem:gfch}
Let $x \in \preals$ and \gfch[x,v] be the \gfcha. We have 
\begin{enumerate}
    \item $\mff[\gfch[x,v][-]] = 6.1x$;
    \item $\{ \load[v] \mid \sol \text{ is an \optsoln of } \mff[\gfch[x,v][-]]\} = \{0, x\}$;
    \item for every \optsoln of \mff[\gfch[x,v][+]] we have $\gen[v] > 0 \implies \gen[g] + \gen[e] + \gen[t] < 6.1x$.
\end{enumerate}
\end{lemma}

The following encodings sum up multiple GCNs together with a ``glue network''.
With the exception of the $v$ nodes, 
GCNs do not intersect one another or the glue network.  
We now use the properties of the \gfchas to show NP-hardness by reduction from the exact cover by 3-set problem. 
This problem is purely combinatorial and hence strongly NP-complete, which implies that there is no fully polynomial time approximation scheme \cite{michael1979computers}.

\begin{theorem}
    \label{theo:strongly}
The \mffa problem is strongly NP-complete.
\end{theorem}
\begin{proof}
    We prove this by reduction from the \emph{exact cover by 3-set problem}. Given a set $M$ and a set of subsets $S \subseteq \powerset{M}$ where every element of $S$ has exactly 3 elements, decide if there exists a set $T \subseteq S$ such that $\bigcup_{X \in T} X = M$ and $\forall X_1, X_2 \in T: X_1 \neq X_2 \implies X_1 \cap X_2 = \emptyset$.

For an instance $(M,S)$, we define the \fdca $\net[*][M,S] := (\buses, \generators, \loads, \lines)$ with $\generators := \{g\}$, $\loads := \{l\}$, $\buses := \generators \cup \loads \cup \bigcup_{X \in S} \{v_X\} \cup \bigcup_{x \in M} \{x\}$ and 
$\lines := \{\edge[g][l][1][3]\} \cup \bigcup_{X \in S} \bigcup_{x \in X} \{\edge[v_X][x][1][1]\}) \cup \bigcup_{x \in M} \{ \edge[g][x][1][1], \edge[x][l][1][2]\}.$
We then define $\net[][M,S] = \net[*][M,S] + \sum_{X \in S} \gfch[3, v_X]$ and we have:
$\mff[\net[][M,S]] = 3 + 18.3|S| + |M| \iff (M,S) \text{ is solvable}.$

An example encoding for the exact cover problem $(M,S) = (\{a,b,c,d,e,f\},\allowbreak \{\{a,b,c\},\allowbreak \{b,c,d\},\allowbreak \{d,e,f\}\})$ can be found in Figure~\ref{exact_cover} and an \optsoln in Figure~\ref{exact_cover_mf}.
\pics[Example encoding for $(M,S) = (\{a,b,c,d,e,f\}, \{\{a,b,c\}, \{b,c,d\},\allowbreak \{d,e,f\}\})$.]{
	\subpic{The network \net[][M,S]}{exact_cover}
	\subpic{A solution of \net[][M,S]}{exact_cover_mf}
}[pic:strongly]
\qed
\end{proof}
We can use almost the same proof for the \msfa: it suffices to replace \gfcha with \gscha and changing the number $18.3$ to $9$.

We now present an alternative reduction for the switching configuration problem 
that provides a stronger result: switching for planar networks is strongly NP-complete.
We use the Hamiltonian Path problem which, given two nodes $a$ and $b$ in a graph, 
consists in deciding whether there exists a path that starts in $a$, ends in $b$ and visits every node exactly once.
The Hamiltonian circuit problem is strongly NP-complete even for planar and cubic graphs \cite{garey1976planar} 
and the generalization to the Hamiltonian Path problem is trivial.

\begin{theorem}
\label{theo:hamilton}
The \msfa problem for planar networks with max degree of $3$ is strongly NP-complete.
\end{theorem}
\begin{proof}
We prove that by reduction from the $a-b$ Hamiltonian Path problem. 
Let $(\buses[h], \lines[h])$ be an arbitrary graph with $a, b \in \buses[h]$ and $\buses[h] = \{v_1,\ldots v_n\}$ with $a = v_1$ and $b = v_n$. 
We define additional nodes $v_0, v_{n+1}$ and $v_i'$ where $0 \leq i \leq n+1$ and we set $v_0 = v_0'$ and $v_{n+1} = v_{n+1}'$.
The encoding of $(\buses[h], \lines[h])$ into an \dca is $\net[][\buses[h], \lines[h], a, b] := (\buses, \generators, \loads, \lines)$ with $\generators := \{v_0\}$, $\loads := \{v_{n+1}\}$, $\buses := \buses[h] \cup \{ v_i' \mid 0 \leq i \leq n+1 \}$ and $\lines := \{ \edge[c][d][1][1] \mid \edge[c][d] \in \lines[h]\} \cup \{\edge[v_0][v_1][1][1], \edge[v_n][v_{n+1}][1][1]\} \cup \{ \edge[v_i'][v_{i+1}'][1][1] \mid 0 \leq v \leq n\}$.
We have: 
$\msf[\net[][\buses[h], \lines[h], a, b]] = 2 \iff (\buses[h], \lines[h], a, b)$ has a Hamiltonian Path from $a$ to $b$.

Figure~\ref{pic:hamilton} shows an example encoding for a graph with four nodes.
\pics[Example for the Graph $(N,E) = (\{a,c,d,b\},\{\edge[a][c],\allowbreak \edge[a][d], \allowbreak \edge[c][d], \allowbreak \edge[c][b],\allowbreak  \edge[d][b]\}).$]{
	\subpic{The network \net[][N,E]}{hamilton}
	\subpic{A solution for \net[][N,E]}{hamilton_msf}
}[pic:hamilton]
\qed
\end{proof}
Note that the proof uses the same capacity and the same susceptance value for all edges 
(here $1$, but these values can be chosen arbitrarily).
Hence, the complexity of the problem is not so much due to the interaction between different parameter values 
as in the Subset Sum problem we use in the next section,
but strictly results from the discrete aspects introduced by allowing for switching edges.
In fact, we need edge capacities only at the edges of the generator; 
the reduction would still be valid if all other edges were allowed unlimited capacity.

\section{Beyond Trees}
\label{sec:beyond_trees}
As shown in Section~\ref{sec:general_complexity}, the \mffa and \msfa problems are NP-complete 
and cannot be arbitrarily closely approximated in polynomial time in general.
But real world power networks are not arbitrary graphs: 
for instance, their maximum node degree is limited and they are (almost) planar networks, 
hence the class of real world power networks could still be easy to configure optimally.  
Consequently, we study the complexity of \msfa and \mffa problems 
for restricted classes of graphs.  

We start with a tree structure, which is overly restrictive since
real networks need to be meshed to be able to sustain failures.  
The reconfiguration problems are easy on tree networks.  
\begin{lemma}
    \label{lem:tree}
The \msfa and the \mffa problem for tree networks can be solved in polynomial time.
\end{lemma}
\begin{proof}
This is a consequence of the fact that, in the absence of cycles, there are no cyclic dependencies on the phase angles. 
Hence they can be chosen in such a way as to match any \optsoln of the traditional max flow.
\qed
\end{proof}

Because the complexity of our problems is driven by cycles, cacti are an obvious relaxation of trees to study.
They allow for cycles but every edge can only put constraints on at most one cycle.
Both problems turn out to be NP-complete for cacti.

\definecolor{two}{HTML}{40E0D0}
\definecolor{fuchs}{HTML}{FF00FF}
\definecolor{olive}{HTML}{808000}
\definecolor{orang}{HTML}{FFA500}
\definecolor{brow}{HTML}{A52A2A}	
\begin{theorem}
    \label{theo:cactusmsf}
The \msfa problem for cactus networks is NP-complete.
\end{theorem}
\begin{proof}
The proof is by reduction from the \emph{Subset sum problem}. 
Given a finite set $M \subset \naturalnumbers$ and a $w \in \naturalnumbers$:
the {Subset sum problem} is to decide whether there is a $V \subseteq M$ such that $\sum_{x \in V} x = w$. 
If such a subset exists then we call the problem $(M,w)$ \emph{solvable}.

Let $(M,w)$ be an instance of the Subset sum problem with $M = \{ x_1, \ldots x_n \}$ and let $m := \sum_{x \in M} x$.
To encode the problem we use the network $\net[*][M,w] := (\buses, \generators, \loads, \lines)$ where $\generators := \{g\}$, $\loads := \{l\}$, $\buses := \generators \cup \loads \cup \bigcup_{1 \leq i \leq n} \{v_i\}$ and 
$\lines := \{ \edge[g][l][1][2+w], \edge[g][v_1][1][1], \edge[v_1][l][1][w+1]\} \cup \bigcup_{1 \leq i < n} \{ \edge[v_i][v_{i+1}][1][w] \}.$
We define $\net[M,w] := \net[*][M,w] + \sum_{1 \leq i \leq n} \gsch[x_i,v_i]$ and we have: 
$\msf[\net[M,w]] = 3 + w + 3m \iff (M,w) \text{ is solvable}$.

\pics[Example for $(M,w) = (\{\textcolor{olive}{1}, \textcolor{two}{2}, \textcolor{fuchs}{3} \}, \textcolor{orang}{5})$]{
	\subpic{The network \net[][M,w]}{cactus2}
	\subpic{A solution for \net[][M,w]}{cactus2_mf}
}[pic:cactus]
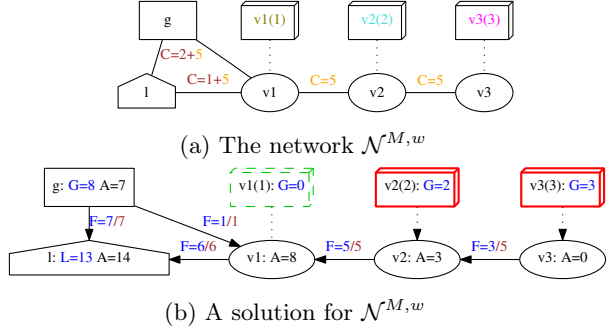
Figure~\ref{pic:cactus} shows an example encoding for $(M,w) = (\{1,2,3\}, 5)$ where the dotted edge symbolizes that the node $v_i$ is the same as the node in the black box of \gsch[x_i,v_i].
\qed
\end{proof}

The same proof can be used for the \mffa problem by replacing \gsch[x_i,v_i] with \gfch[x_i,v_i] and changing $3m$ to $6.1m$.

\begin{theorem}
The \mffa problem for cactus networks is NP-complete.
\end{theorem}

Note that both encodings have a fixed maximum degree. 
In the \msfa proof, the maximum degree of the encoding is $3$ as it only consists of triangles that are connected via a path.
In the case of the \mffa, the maximum degree of $4$ is at the node $c$ of the \gfchn.
However, this can be further reduced to $3$ by splitting $c$ into two nodes that are connected via an edge. 

N-level Tree networks are another possible relaxation of trees.
An n-level Tree network is an \dca based on a tree where there is one generator at the root and loads at the leaves.
Edges that are not part of the tree can only be added between nodes on the same tree level where the level is less or equal to $n$ and only such that the resulting graph is planar.
This network structure is motivated by the disaster management application.
After the destruction of many power lines, it is easier to first repair lines such that we obtain a tree structure.
Then we can start restoring additional edges.

\begin{definition}[n-level Tree network]
	Let $n \in \naturalnumbers$. 
    An \emph{n-level Tree network} is an \dca iff there exists a sub-network $T$ that is a Tree such that:
    all leaves of $T$ are loads;
    there is only one generator at the root node of $T$ and
    there is a total order on the children of every node (which implies a total order on all nodes in one level) such that
    every node of the same tree level can only be connected to its neighbours in the total order on all nodes of the same level.
\end{definition}


\begin{theorem}
    \label{theo:tree}
The \msfa problem for 2-level Tree networks is NP-complete.
\end{theorem}
\begin{proof}
    We prove Theorem~\ref{theo:tree} by reduction from a version of the \emph{subset sum} problem.
    Given an instance $(M, w)$, let $m + 1 := \sum_{x \in M} x$ and $M = \{a_2, \ldots, a_n\}$.
    We use $a_i$ to represent a value form $M$ as well as a symbol corresponding to that value.
    The network \net[][M,w] is defined by $\net[][M,w] := (\buses, \generators, \loads, \lines)$ with $\generators := \{g\}$, $\loads := \{ l_i \mid 1 \leq i \leq n+1 \}$, $\buses := \{g_1, g_{n+1}, p, a_1, a_{n+1} \} \cup M \cup \generators \cup \loads$ and 
    $\lines := 
    \bigcup_{1 \leq i \leq n} \allowbreak\{ \edge[p][a_i][\frac{a_i}{i}][a_i],\allowbreak \edge[a_i][l_i][1][a_i], \allowbreak\edge[a_{i-1}][a_i][m][m] \} 
    \cup
    \{ \edge[g][g_1][2m+2][m+1], \edge[g_1][a_1][2m+2][m+1],\allowbreak \edge[a_1][l_1][1][1], \edge[g][p][w][w], \edge[g][g_{n+1}][\frac{2}{n+1}][1],\allowbreak \edge[g_{n+1}][a_{n+1}][\frac{2}{n+1}][1], \edge[a_{n+1}][l_{n+1}][1][m + 1], \edge[a_n][a_{n+1}][m][m] \}$.
    We have $\msf[\net[][M,w]] = m + 2 + w \iff (M,w) \text{ is solvable}$.

    An example encoding for $(M,w) = (\{2,1,3\}, 5)$ can be found in Figure~\ref{pic:tree}.
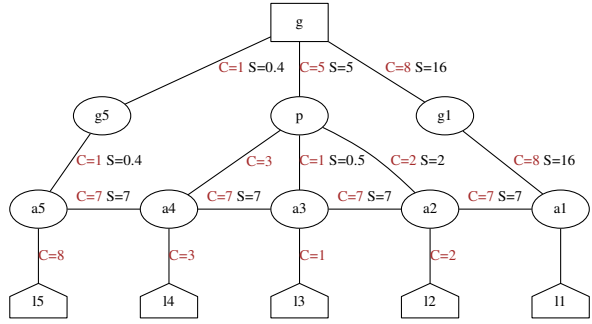
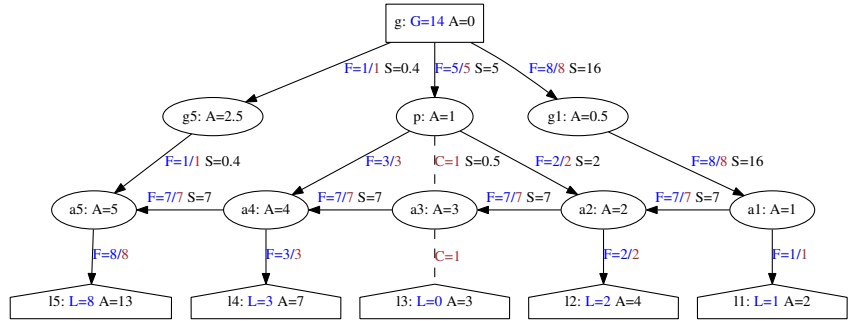
\pics[Example for $(M,w) = (\{2,1,3\}, 5)$]{
	\subpic{The network \net[][M,w]}{treelevel}
	\subpic{A solution for \net[][M,w]}{treelevel_msf}
}[pic:tree]
\qed
\end{proof}

The two classes of graph we just studied, 
cacti and n-level Tree networks, 
are much more constrained than real world power networks.  
From the results of this section, 
we conclude that the reconfiguration problems 
are also hard for real world networks.

\section{Conclusion}
\label{conclusion}
This paper provides complexity results for two reconfiguration max flow problems in power systems: switching and utilizing FACTS devices.
It shows that the switching problem is hard even for simple network structures and that we cannot expect an arbitrarily close approximation in polynomial time.
It also shows that the problem of maximum flow with FACTS devices is hard for simple network structures and strongly NP-complete in general.

There remain some open questions in this setting, such as whether the
optimal solution can be efficiently approximated on cacti or real
world power networks arbitrarily closely or within a constant factor.
Additionally, to help find good solutions, exploitable properties
and/or heuristics must be investigated.

Real network operations however 
require further guarantees on the robustness of the configuration.  
A criterion widely used in the industry is the $N-1$ property
which states that a single failure on the network (i.e., removal of a line) 
will not cascade into a major blackout.  
Whether and how the $N-1$ requirement conflicts 
with the optimization objective remains unknown.  

A last issue is the transitional aspect of reconfiguration.  
We only assumed steady-state situations in this paper 
but transiting from one configuration to the next 
(for instance, switching a line) is a complex process, 
sometimes even infeasible.  
This adds a layer of complexity to the problem of reconfiguration 
that must be considered in real applications.

\bibliographystyle{splncs}
\bibliography{braess_paradox}
\newpage
\begin{appendix}
\section{Proofs}

%
%
\subsubsection{Proof of Theorem~\ref{theo:hamilton}:}
It is easy to see that the reduction is polynomial.

To achieve a \msfa of $2$, all edges of $v_0$ and $v_{n+1}$ have to be congested. To achieve the congestion of the path $\{v'_0, \ldots, v'_{n+1}\}$, we have to have that $\pa[v'_{n+1}] - \pa[v'_0] = n + 1$. The fact that the edges \edge[v_0][a] and \edge[b][v_{n+1}] are congested implies that $\pa[b] - \pa[a] = n - 1$. 

Let $a = c[1], \ldots, c[k] = b$ be a path of length $k$ in the graph $(\buses[h], \lines[h])$ from $a$ to $b$. The DC power law implies that every edge in that graph allows for a maximum phase angle difference of $1$. Therefore, the maximum phase angle difference between $a$ and $b$ with respect to this path is bounded by $k-1$, so $\pa[b]- \pa[a] \leq k-1$.

If the \msfa is $2$, then because $\pa[b] - \pa[a] = n-1 \leq k-1$ and $k \leq n$, there has to be a path of length $n$. On the other hand, if there exists a path of length $n$, then we simply switch all edges that are not part of this path. Given a phase angle difference of $n-1$ between $a$ and $b$, this path has a flow of $1$ and hence Kirchhoff's conservation law is fulfilled for $a$ and $b$. This shows that there is a solution with an \msfa of $2$.

%
%
\subsubsection{Proof of Theorem~\ref{lem:gsch}:}
    To achieve $\gen[g] = 3x$ we need that both edges \edge[g][l][1][2x] and \edge[g][v][1][x] are unswitched and congested.
    One way is to have $\flow[v][l][1][x] = x$ in which case $\load[v] = 0$.
    On the other hand we can switch \edge[v][l][1][x] which implies $\load[v] = x$.

    If $v$ acted as a generator the flow on the edge \edge[g][v][1][x] would be reduced and hence the generation on $g$.
    
%
%
\subsubsection{Proof of Lemma~\ref{lem:gfch}:}
    Note that to simplify notations we omit writing the susceptances and capacities on our edges.
    Table~\ref{tab:cn} presents two flows with a generation of $6.1x$. 
    This implies that every \optsoln has to have a generation greater or equal to $6.1x$.
    
    Let \sol be an \optsoln of the \mffa.    

    The capacity of the edge \edge[e][v] is $0.4x$.
    First we show that when the flow goes from $v$ to $e$ then we maximally have a flow of $-0.1x$.
    We have:
    \begin{align*}
    6.1x &\leq \mff[\gfch[x,v][-]] = \gen[g] + \gen[t] + \gen[e]\\
    	&= \flow[g][v] + \flow[t][c] + \flow[t][l] + \flow[e][c] + \flow[e][v]\\ 
    	&\leq \capa[g][v] + \capa[t][c] + \capa[t][l] + \capa[e][c] + \flow[e][v]\\
        &= 6.2x + \flow[e][v]
    \end{align*}
    which implies $-0.1x \leq \flow[e][v]$.

	Next, we show that 
	\begin{equation}
		\label{eq:gfch_upper}
	\flow[e][v]\frac{\sus[e][v] - 1}{\sus[e][v]} \leq 0.15x
	\end{equation}
	 and that equality is true if and only if $\flow[e][v] = 0.4x$ and $\sus[e][v] = 1.6$ or $\flow[e][v] = -0.1x$ and $\sus[e][v] = 0.4$.    

	 Let us assume that \flow[e][v] is positive.
     To obtain an upper bound for $\flow[e][v]\frac{\sus[e][v] - 1}{\sus[e][v]}$ the term $\sus[e][v] - 1$ has to be positive which implies that $1.6 \geq \sus[e][v] \geq 1$ and given that the capacity of \flow[e][v] is $0.4x$ we have $\flow[e][v]\frac{\sus[e][v] - 1}{\sus[e][v]} \leq 0.4x \frac{1.6 - 1}{1.6} = 0.15x$. 
	 
     On the other hand, if \flow[e][v] is negative then $\sus[e][v] - 1$ has to be negative which implies $0.4 \leq \sus[e][v] \leq 1$.
     Because $-0.1x \leq \flow[e][v]$, we have $\flow[e][v]\frac{\sus[e][v] - 1}{\sus[e][v]} \leq -0.1x \frac{0.4 - 1}{0.4} = 0.15x$.\\ 

	Next we show:
	\begin{equation}
		\label{eq:gfchtl}
		\flow[t][l] = 2\flow[t][c] + 2\flow[e][c] - \frac{\flow[e][v]}{\sus[e][v]}.
	\end{equation}	 
	This equation can be derived by combining the following three equations:
    \begin{subequations}
    \begin{align}
        \flow[t][l] = \flow[t][c] + \flow[c][l]\\
		\label{eq:gfchl}
		\flow[e][c] = \frac{\flow[e][v]}{\sus[e][v]} + \flow[v][c]\\
        \flow[c][l] = \flow[t][c] + \flow[v][c] + \flow[e][c]
    \end{align}
    \end{subequations}
    The first two are obtained from the LDC power law applied to the triangles formed by $\{t, l, c\}$ and $\{v,e,c\}$ where the last one comes from Kirchhoff's conservation law at node $c$.

	Finally, if we use Eq.~\ref{eq:gfchtl} to substitute \flow[t][l] and apply Eq.~\ref{eq:gfch_upper} we get:
	\begin{align*}
    6.1x &\leq \mff[\gfch[x,v][-]] = \gen[g] + \gen[t] + \gen[e]\\
	&= \flow[g][v] + \flow[t][c] + \flow[t][l] + \flow[e][c] + \flow[e][v]\\  
	&= \flow[g][v] + 3\flow[t][c] + 3\flow[e][c] + \flow[e][v](1 - \frac{1}{\sus[e][v]})\\
	&\leq \capa[g][v] + 3 \capa[t][c] + 3 \capa[e][c] + 0.15x = 6.1x
	\end{align*}
	which implies $\mff[\gfch[x,v][-]] = 6.1x$, $\flow[g][v] = x$, $\flow[e][c] = 0.65x$, $\flow[t][c] = x$ and $\flow[e][v]\frac{\sus[e][v] - 1}{\sus[e][v]} = 1.5x$ for all solutions. As argued above, there are only two possible value pairs of $(\flow[e][v], \sus[e][v])$ for the last term, namely $(0.4x,1.6)$ and $(-0.1x,0.4)$.
    For $(0.4x,1.6)$ and using Eq.~\ref{eq:gfchl} to compute \flow[v][c], Kirchhoff's conservation law at $v$ gives us 
		$$\load[v] = \flow[g][v] + \flow[e][v] - \flow[v][c] = x + 0.4x - (0.65x - 0.25x) = x.$$
	Similarly, for $(-0.1x, 0.4)$ we have 
		$$\load[v] = \flow[g][v] + \flow[e][v] - \flow[v][c] = x - 0.1x - (0.65x + 0.25x) = 0.$$
		
    We know show that for every \optsoln of \mff[\gfch[x,v][+]] we have $\gen[v] > 0 \implies \gen[g] + \gen[e] + \gen[t] < 6.1x$.
    In the \optsoln where $\load[v] = 0$, all edges from all remaining loads ($l$) are congested.
    Hence it is not possible to increase the inner generation by having $v$ act as generator.
    It is easy to see that $\gen[v] + \gen[g] = x$ has to be true and hence if $\gen[v] > 0$ then $\gen[g] + \gen[e] + \gen[t] < 6.1x$.

%
%
\subsubsection{Proof of Theorem~\ref{theo:strongly}:}
It is easy to see that this encoding is polynomial. 
Let $X \in S$. 
Lemma~\ref{lem:gfch}:1 tells us that \gfch[3,v_X] can maximaly generate $18.3$ if $v_X$ acts as load for \gfch[3,v_X]. 
On the other hand, Lemma~\ref{lem:gfch}:3 tells us that if we do the opposite and provide power to \gfch[3,v_X] via the node $v_X$ then the generation of the network is strictly less than $18.3$.
The sum of capacities of all edges connected to $g$ is $3 + |M|$.
Hence, to achieve an overall generation of $3 + 18.3|S| + |M|$ all edges from the generator have to be congested, all networks \gfch[3,v_X] have to be at their \mffa and every $v_X$ has to act as a generator for the network \net[*][M,S].
The later together with Lemma~\ref{lem:gfch}:2 implies that every $v_X$ either provides $0$ or $3$ to the network \net[*][M,S].
$v_X$ has three edges \edge[v_X][x][1][1] with $x \in X$.
Therefore, there is either no incoming flow to all nodes $x \in X$ from $v_X$ or all gain $1$. We call a network \gfch[3, v_X] that provides $1$ to its $x \in X$ \emph{active}.

For every \optsoln with phase angles \pas we assume w.l.o.g. that $\pa[g] = 0$.
Also, because the edges \edge[g][l][1][3] and \edge[g][x][1][1] have to be congested, we have $\pa[l] = 3$ and $\forall x \in M: \pa[x] = 1$. These phase angles imply that $\flow[\edge[x][l][1][1]] = 2$.
They also imply that there cannot be any flow from one node $x_1 \in M$ to any other node $x_2 \in M$ because they all have the same phase angle.
Since there is an incoming flow of $1$ from the edge \edge[g][x][1][1], to fulfil Kirchhoff's conservation law the node $x$ needs an additional incoming flow of exactly $1$.
Therefore, for all $X \in S$ with $x \in X$ there has to be exactly one active network \gfch[3,v_X].

Let $T$ be a solution of $(M,S)$. We activate all networks \gfch[3, v_X] where $X \in T$. Because $T$ is a solution, this will provide a flow of exactly $1$ to every node $x$ and hence Kirchhoff's conservation law at every node $x \in M$ is satisfied.

On the other hand, let \swsol be a solution of the \mffa with a total generation of $3 + 18.3|S| + |M|$. We define $T := \{ X \in S \mid \gfch[3, v_X] \text{ is active}\}$. The solution satisfies Kirchhoff's conservation law and therefore, using our observations above, every node is connected to exactly one active network. Therefore, $T$ is a solution of $(M,S)$.

%
%
\subsubsection{Proof of Theorem~\ref{theo:cactusmsf}:}
It is easy to see that this reduction is polynomial and because the networks \gsch[x,v] are cactuses, that the network \net[M,w] is a cactus. 

Let $x_i \in M$. 
Lemma~\ref{lem:gsch}:1 tells us that \gsch[x_i,v_i] can maximally generate $3x_i$ if $v_i$ acts as load for \gsch[x_i,v_i]. 
On the other hand, Lemma~\ref{lem:gsch}:3 tells us that if we do the opposite and provide power to \gsch[x_i,v_i] via the node $v_i$ then the generation of the network is strictly less than $3x_i$.
The sum of capacities of all edges connected to $g$ is $3 + w$.
Hence, to achieve an overall generation of $3 + w + 3\sum_{1 \leq i \leq n} x_i$ all edges from the generator have to be congested, all networks \gsch[x_i,v_i] have to be at their \msfa and every $v_i$ has to act as a generator for the network \net[*][M,S].
The later together with Lemma~\ref{lem:gsch}:2 implies that $v_i$ either provides $0$ or $x_i$ to the network \net[*][M,S].
We call a network \gsch[x_i,v_i] that provides $x_i$ \emph{active}.
The fact that all edges from the $g$ are congested implies that the edge \edge[v_1][l][1][w+1] has to be congested as well and flows from $v_1$ to $l$. 
Since $g$ provides only a flow of $1$ to $v_1$, we achieve an \msfa of $3 + w + 3m$ if and only if the networks \gsch[x_i,v_i] provide the other $w$.

Let $V$ be a solution of $(M,w)$. Then we active all networks \gsch[x_i,v_i] with $x_i \in V$. Because $\sum_{x \in V} x = w$ we know $v_1$ fulfils Kirchhoff's conservation law.

On the other hand, if $\msf[\net[M,w]] = 3 + w + 3m$, then we take any \optsoln and define $V := \{ x_i \in V \mid \gsch[x_i,v_i] \text{ is active}\}$. Since we know that the solution fulfils Kirchhoff's junction law in $v_1$ we have $\sum_{x \in V} x = w$.

%
%
\subsubsection{Proof of Theorem~\ref{theo:tree}:}
	It is easy to see that \net[M,w] is a 2-level Tree network. 

	\textbf{Case 1:}$\msf[\net[M,w]] = m + 2 + w \implies (M,w) \text{ is solvable}.$\\
	Let \dsol be an \optsoln of the \msfa. 
    W.l.o.g. let $\pa[g] = 0$. 
    Since the max flow is $m + 2 + w$ we know that the edges \edge[g][g_1], \edge[g_1][a_1], \edge[g][g_{n+1}], \edge[g_{n+1}][a_{n+1}] and \edge[g][p] are congested and therefore $\pa[a_1] = \pa[p] = 1$ and $\pa[a_{n+1}] = 2 \frac{n+1}{2} = n+1$. 
    This implies that we have at least $m + 2$ incoming power at the node $a_1$. 
    Since the other two edges have in sum a capacity of $m + 2$ we know that they are congested. 
    Therefore, we obtain $\pa[a_2] = 2$. 
    For the node $a_n$ we know that the phase angle can not be bigger then $n$ because that would overload the edge \edge[p][a_i][a_i][\frac{a_i}{i}]. 
    However, if the phase angle is smaller then $n$, then the edge \edge[a_n][a_{n+1}][m][m] is overloaded. 
    Therefore, $\pa[a_n] = n$. 
    We can apply similar arguments to $a_{n-1}$ to have a phase angle of $n-1$. 
    Overall, we derive that $\forall 1 \leq i \leq n: \pa[a_i] = i$. 
    Hence, the edges \edge[a_{i-1}][a_i][m][m] for $1 \leq i \leq n+1$ must be congested. 
    They also cannot be in \linesw since their flow $m$ is greater then the sum of elements of $M$ an therefore the sum of power we can send to the loads $l_i$ and can get from the nodes $a_i$  with $1 \leq i \leq n$.
		
	We define $V := \{ a_i \in M \mid \edge[p][a_i] \notin \linesw \}.$ We know that the incoming power at node $p$ is $w$, $p$ respects Kirchhoff's conservation law and that all edges $\edge[p][a_i] \notin \linesw$ are congested. Therefore $\sum_{x \in M} x = w$.\\
	\newline
	\textbf{Case 2:}$(M,w) \text{ is solvable} \implies \msf[\net[M,w]] = m + 2 + w.$\\	
	Let $V$ be a solution of $(M,w)$. We define phase angle \pas with $\pa[g] := 0, \pa[g_1] := \frac{1}{2}, \pa[g_{n+1}] := \frac{n+1}{2}, \pa[a_i] := i, \pa[l_1] := 2, \pa[l_{n+1}] ;= n+2+m$ and 	$\forall 1 \leq i \leq n:$
	$$\pa[l_i] := 
		\begin{cases}
			i + a_i		&\text{if } a_i \in V\\
			i			&\text{otherwise}.
		\end{cases}
	$$
	We also define $\linesw := \{ \edge[p][a_i] \in \lines \mid a_i \notin V \}$. Since the sum of all elements of $V$ is $w$, we know that with this definition $p$ respects Kirchhoff's conservation law and it is easy to see that all other nodes do the same. This definition also implies a flow of $m + 2 + w$ which is the \msfa because all edges of the generator are congested.
\end{appendix}

\end{document}